%% file: SINRBroadcastWithNoise-arxiv-02-2013.tex
\documentclass[letterpaper,11pt]{article}
\usepackage[margin=0.97in]{geometry}
\usepackage{amsthm,amsmath}
\usepackage{amssymb}

\usepackage[ruled]{algorithm}
\usepackage[noend]{algpseudocode}

\newif\iffull
\fullfalse

\usepackage{epsfig}

\renewcommand{\thefootnote}{\fnsymbol{footnote}}

\newtheorem{theorem}{Theorem}
\newtheorem{corollary}{Corollary}
\newtheorem{lemma}{Lemma}

\newtheorem{fact}{Fact}

\newtheoremstyle{redstyle}
     {3pt}
     {3pt}
     {\color{black}}
     {}
     {\color{red}\bfseries}
     {:}
     {.5em}
     {}
\theoremstyle{redstyle}

\newcommand{\m}{\mathcal}
\newcommand{\cT}{{\mathcal T}}
\newcommand{\cN}{{\mathcal N}}

\newcommand{\cI}{{\mathcal I}}
\newcommand{\ep}{\varepsilon}
\newcommand{\eps}{\varepsilon}

\newcommand{\DIR}{\text{DIR}}

\newcommand{\dogory}{\vspace*{-3pt}}

\newcommand{\labell}[1]{\label{#1}}
\renewcommand{\paragraph}[1]{\vspace*{0.5ex}\noindent {\bf #1}}
\newcommand{\remove}[1]{}

\newcommand{\dist}{\text{dist}}

\newcommand{\NAT}{{\mathbb N}}
\newcommand{\INT}{{\mathbb Z}}

\usepackage{color}
\newcommand{\dk}[1]{#1} 
\newcommand{\tj}[1]{#1} 
\newcommand{\comment}[1]{}

\newcommand{\degree}{^o}

\begin{document}

\title{Distributed Broadcasting in Wireless Networks under the SINR Model\thanks{%
	This work was supported by the
	EPSRC grant EP/G023018/1.}
	}

\author{%
    Tomasz Jurdzinski\footnotemark[2] \footnotemark[3]
    \and
    Dariusz R.~Kowalski\footnotemark[3] \footnotemark[4]
    \and
    Tomasz Maciejewski\footnotemark[2]
    \and
    Grzegorz Stachowiak\footnotemark[2]
}

\footnotetext[2]{Institute of Computer Science, University of Wroc{\l}aw, Poland.}

\footnotetext[3]{Department of Computer Science,
            University of Liverpool,
            Liverpool L69 3BX, UK.
            }

\footnotetext[4]{Institute IMDEA Networks, 28918 Madrid, Spain.}

\date{}

\maketitle


\begin{abstract}
In the advent of large-scale multi-hop wireless technologies, such as MANET, VANET, iThings, it is of utmost importance to devise efficient distributed protocols to maintain network architecture and provide basic communication tools. One of such fundamental communication tasks is broadcast, also known as a 1-to-all communication.
We propose several new efficient distributed algorithms and evaluate their time performance both theoretically and by simulations. First randomized algorithm accomplishes broadcast in $O(D+\log(1/\delta))$
rounds with probability at least $1-\delta$ on {\em any} uniform-power network of $n$ nodes and diameter $D$, when equipped with local estimate of network density.
Additionally, we evaluate average performance of this protocols by simulations on two classes of
generated networks --- uniform and social --- and compare the results with
performance of exponential backoff heuristic.
Ours is the first provably efficient and well-scalable distributed solution for the (global) broadcast task.
The second randomized protocol developed in this paper
does not rely on the estimate of local density, and achieves only slightly higher time performance
$O((D+\log(1/\delta))\log n)$.

\vspace*{1ex}

\end{abstract}



\renewcommand{\thefootnote}{\arabic{footnote}}


\section{Introduction}

In this work we consider a broadcast problem in ad-hoc wireless networks under the
Signal-to-Interference-and-Noise-Ratio model (SINR).
Wireless network consists of at most $n$ stations, also called nodes, with unique integer IDs
and uniform transmission powers $P$, deployed in the two-dimensional space with Euclidean metric.
Each station initially knows only its own ID, location and
the upper bound $n$ on the number of nodes.
Locations of stations and parameters of the SINR model determine a {\em communication graph}
of the network:
It is defined on network nodes
and contains links $(v,w)$ such that the distance from $v$ to $w$ is at most $(1-\eps)$ of the
maximum SINR ratio possible at node $w$,
where $0<\eps<1$ is a fixed model parameter.
\dk{This definition is common in the literature, c.f., \cite{YuHWTL12}, and naturally motivated.
Intuitively, we should be able to propagate messages along the links in the communication graph,
but not necessarily between any two non-connected nodes, as the distance between them
(bigger than $(1-\eps)$ fraction of the maximum SINR ratio possible) could be easily interrupted by even
a single faraway transmission in realistic scenarios.}
%
%
%
We consider two settings: one with local knowledge of density, in which each station knows also the
number of other stations in its close proximity (dependent on parameter $\eps$) and the other
when no extra knowledge is assumed.

In the broadcast problem, there is one designated node, called a source, which has a piece
of information, called a source message or a broadcast message,
which must be delivered to all other accessible nodes by using wireless communication.
In the beginning, only the source is executing the broadcast protocol,
and the other nodes join the execution after receiving the broadcast message for the first time.
The goal is to minimize time needed for accomplishing the broadcast~task.

\subsection{Previous and Related Results}

In this work, we study
the problem of {\em distributed broadcasting} in ad hoc wireless networks
under the SINR physical model, from both theoretical and simulation perspectives.
%
%
In what follows, we discuss most relevant results in the SINR model, and
the state of the art obtained in the older Radio Network model.

\paragraph{SINR model.}
In the 
SINR model in ad hoc setting,
slightly weaker task of {\em local} broadcasting, in which nodes have to inform
only their neighbors in the corresponding communication graph,
was studied in \cite{YuWHL11}.
The considered setting allowed power control by deterministic algorithms,
in which, in order to avoid collisions,
stations could transmit with any power smaller than the maximal one.
Randomized solutions for contention resolution~\cite{KV10}
and local broadcasting~\cite{GoussevskaiaMW08} were also obtained.

Recently, a distributed randomized algorithm for multi-broadcast has been
presented \cite{YuHWTL12} for uniform networks. Although the problem solved in that paper
is a generalization
of broadcast, the presented solution is restricted merely to networks having
the communication graph connected for $\eps=\frac23 r$, where $r$
is the largest possible SINR ratio. In contrast, our solutions are efficient and scalable
for {\em any} networks with communication graph connected for {\em any} value of $\eps<\frac{1}{2}$.

There is a vast amount of work on centralized algorithms under the SINR model.
The most studied problems include connectivity, capacity maximization,
link scheduling types of problems;
for recent results and references we refer the reader to the survey~\cite{WatSurv}.
Multiple Access Channel properties were also recently studied
under the SINR model, c.f.,~\cite{RichaSSZ}.

\paragraph{Radio network model.}
There are several papers analyzing broadcasting in the radio model of wireless networks,
under which a message is successfully heard if there are no other simultaneous transmissions
from the {\em neighbors} of the receiver in the communication graph.
This model does not take into account the real strength of the received signals, and also the signals
from outside of some close proximity.
In the {\em geometric} ad hoc setting, Dessmark and Pelc~\cite{DessmarkP07} were the first who studied
the broadcast problem.
They analyzed the impact of local knowledge, defined as a range within which
stations can discover the nearby stations.
Unlike most research on broadcasting problem and the assumptions of this paper,
Dessmark et\ al. \cite{DessmarkP07} assume spontaneous wake-up of stations. That is,
stations are allowed to do some pre-processing
(including sending/receiving messages) prior receiving
the broadcast message for the first time.
Moreover it is assumed in \cite{DessmarkP07} that IDs are strictly from $\{1,\ldots,n\}$,
which makes the setting even
less comparable with the one considered in this work.
Emek et al.~\cite{EmekGKPPS09} designed a broadcast algorithm
working in time $O(Dg)$
in UDG radio networks with eccentricity $D$ and granularity $g$, where
eccentricity was defined as the minimum number of hops to propagate the broadcast message throughout
the whole network and
granularity was defined as
the
inverse of the minimum distance between any two stations times the \tj{maximal range} of a station.
Later, Emek et al.~\cite{EmekKP08} developed a matching lower bound $\Omega(Dg)$.
There were several works analyzing deterministic broadcasting in geometric graphs in the centralized radio setting,
c.f.,~\cite{GasieniecKLW08,SenH96}.

The problem of broadcasting is well-studied in the setting of {\em graph radio model}, in which stations
are not necessarily deployed in a metric space;
here we restrict to only the most relevant results.
In deterministic ad hoc setting with no local knowledge, the fastest $O(n\log(n/D))$-time algorithm in symmetric networks was developed by Kowalski~\cite{Kow-PODC-05}, and almost matching lower
bound was given by Kowalski and Pelc~\cite{KP-DC-05}.
For recent results and references in less related settings we refer the reader
to~\cite{DeMarco-SICOMP-10, KP-DC-05,Censor-HillelGKLN11}
%
There is also a vast literature on randomized algorithms for broadcasting in graph radio model
\cite{KushilevitzM98,KP-DC-05,CzumajRytter-FOCS-03}.
Since they are quite efficient, there are very few studies of the problem restricted to
the geometric setting. However, when mobility of stations is assumed, location and movement
of stations on the plane is natural. Such settings were studied e.g.,\ in
\cite{Farach-ColtonM07}.


\subsection{Our Results}
In this paper we present distributed algorithms for broadcasting in wireless connected networks
deployed in two dimensional Euclidean space under the SINR model, with uniform power assignment
and any $\ep<\frac{1}{2}$.
We distinguish between the two settings:
one with local knowledge of density, in which each station knows the
\tj{upper bound on the number of other stations in its close proximity (dependent on parameter $\eps$)}
and the other
when no extra knowledge is assumed.

\dk{In the former model, we develop a randomized broadcasting
algorithm with time complexity $O(D+\log (1/\delta))$, where $D$ is the eccentricity of
the communication graph, and $\delta$ is the maximal error probability.
This analysis is complemented by the results of simulations on uniform and social networks,
which compare favorably with the performance of exponential backoff protocol.
In the latter model, we give a solution with time complexity $O((D+\log (1/\delta))\log n)$.
All these results hold for model parameter $\alpha>2$;
for $\alpha=2$ the randomized solutions are slower by factor $\log^2 n$ and the deterministic one
becomes slower as well.
}

\vspace*{-1ex}
\section{Model, Notation and Technical Preliminaries}

Throughout the paper, $\NAT$ denotes the set of natural numbers,
$\NAT_+$ denotes the set $\NAT\setminus\{0\}$, and $\INT$
denotes the set of integers.
For $i,j\in\INT$, we use the notation $[i,j]=\{k\in\NAT\,|\,i\leq k\leq j\}$
and $[i]=[1,i]$.

We consider a wireless network consisting of $n$ {\em stations}, also called {\em nodes},
deployed into a two dimensional
Euclidean space and communicating by a wireless medium.
All stations have unique integer IDs in set $[\cI]$; in this paper, we assume that
$\cI=\text{poly}(n)$.
Stations of a network are denoted by letters $u, v, w$, which simultaneously
denote their IDs.
Stations are located on the plane with {\em Euclidean metric} $\dist(\cdot,\cdot)$,
and each station knows its coordinates.
Each station $v$ has its {\em fixed transmission power} $P_v$, which is a positive real number;
in each round, each station either does not transmit a message or it transmits with
its full transmission power $P_v$.
In this work we consider a uniform transmission power setting in which $P_v=1$ for every station $v$.
There are three fixed model parameters: path loss
$\alpha\geq 2$,
threshold $\beta\ge 1$, and ambient noise $\cN\ge 1$.
The $SINR(v,u,\cT)$ ratio, for given stations $u,v$ and a set of (transmitting) stations $\cT$,
is defined as follows:
\vspace*{-1ex}
\begin{equation}\label{e:sinr}
SINR(v,u,\cT)
=
\frac{P_v\dist(v,u)^{-\alpha}}{\cN+\sum_{w\in\cT\setminus\{v\}}P_w\dist(w,u)^{-\alpha}}
\end{equation}

\vspace*{-1ex}
In the {\em Signal-to-Interference-and-Noise-Ratio model} (SINR) considered in this work,
station $u$ successfully receives a message from station $v$ in a round if
$v\in \cT$, $u\notin \cT$, and

\vspace*{-2.5ex}
$$SINR(v,u,\cT)\ge\beta \ ,$$

\vspace*{-1.5ex}
\noindent
where $\cT$ is the set of stations transmitting at that round.


\paragraph{Ranges and uniformity.}
The {\em communication range} $r_v$ of a station $v$ is the radius of the circle in which a message transmitted
by the station is heard, provided no other station transmits at the same time. That is $r_v$ is the largest
value such that $SINR(v,u,\cT)\geq \beta$, provided $\cT=\{v\}$ and $d(v,u)=r_v$.
A network
is
{\em uniform}, when ranges (and thus transmission powers) of all stations are equal,
or {\em nonuniform} otherwise.
In this paper, only uniform networks are considered.
For clarity of presentation
we make the assumption that all powers are equal, 
i.e., $P_v=P$ for each $v$.
Thus, $r_v=r$ for $r=\left(\frac{P}{\beta\cN}\right)^{1/\alpha}$ and each station $v$.
\tj{For simplicity,} we assume that $r=1$ which implies that $P=\beta \cN$.
The assumption that $r=1$ can be dropped without changing
asymptotic formulas for presented algorithms and lower bounds.

\remove{
Under these assumptions, $r_v=r=(1+\eps)^{-1/\alpha}$
for each station $v$.
The {\em range area} of a station with range $r$
located at the point $(x,y)$ is defined as the circle with radius $r$.
}



\paragraph{Communication graph and graph notation.}
The {\em communication graph} $G(V,E)$
of a given network
consists of all network nodes and edges $(v,u)$ such that $d(v,u)\leq (1-\eps)r=1-\eps$,
\tj{where $\eps<1$} is
a fixed model parameter.
The meaning of the communication graph is as follows: even though the idealistic communication
range is $r$, it may be reached only in a very unrealistic case of single transmission in the whole
network. In practice, however, many nodes located in different parts of the network often
transmit simultaneously, and therefore it is reasonable to assume that we may
hope for a slightly smaller range to be achieved.
The communication graph envisions the network of such ``reasonable reachability''.
Observe that the communication graph is symmetric for uniform networks, which are considered
in this paper.
By a {\em neighborhood} of a node $u$ we mean the set (and positions) of all
neighbors of $u$ in the communication graph $G(V,E)$
of the underlying network, i.e., the set $\{w\,|\, (w,u)\in E\}$.
The {\em graph distance} from $v$ to $w$ is equal to the length of a shortest path from $v$ to $w$
in the communication graph, where the length of a path is equal to the number of its edges.
The {\em eccentricity} of a node
is the maximum graph
distance from this node to all other nodes
(note that the eccentricity is of order of the diameter if the communication
graph is symmetric --- this is also the case in this work).

\tj{
We say that a station $v$ transmits {\em $c$-successfully} in a round
$t$ if $v$ transmits a message in round $t$ and this message is received by
each station $u$ in Euclidean distance from $v$ smaller or equal to $c$. }
We say that node $v$ transmits {\em successfully} to node $u$ in a round $t$ if $v$ transmits
a message in round $t$ and $u$ receives this message.
A station $v$ transmits {\em successfully} in round $t$ if it transmits
successfully to each of its neighbors in the communication graph.
%
%

\paragraph{Synchronization.}
It is assumed that algorithms work synchronously in time slots, also called {\em rounds}:
each station can act
either as a sender or as a receiver during a round.
We do not assume global clock ticking; algorithm could easily synchronize their rounds by updating
round counter and passing it along the network with messages.

\paragraph{Collision detection.}
We consider the model without {\em collision detection}, that is,
if a station $u$ does not receive a message in a round $t$, it has no information
whether any other station was transmitting
in that round
and about the value of $SINR(v,u,\m{T})$, for any station $v$, where $\m{T}$ is the set
of transmitting stations in round $t$.

\paragraph{Broadcast problem and performance parameters.}
In the broadcast problem studied in this work, there is one distinguished node, called a {\em source},
which initially holds a piece of information, also called a {\em source message} or a
{\em broadcast message}.
The goal is to disseminate this message to all other nodes by sending messages along the network.
The complexity measure is the worst-case time to accomplish the broadcast task,
taken over all connected networks with specified parameters.
Time, also called the {\em round complexity}, denotes here the number of communication rounds in
the execution of a protocol: from the round when the source is activated
with its broadcast message till the broadcast task is accomplished (and each station is aware that
its activity in the algorithm is finished).
For the sake of complexity formulas, we consider the following parameters:
$n$, $N$, $D$, and $g$, where
$n$ is the number of nodes,
$[N]$ is the range of IDs,
$D$ is the eccentricity of the source,
and $g$ is the granularity of the network, defined as $r$ times the
inverse of the minimum distance between any two stations (c.f.,~\cite{EmekGKPPS09}).

\paragraph{Messages and initialization of stations other than source.}
We assume that a single message sent in the execution of any algorithm
can carry the broadcast message and at most polynomial in the
size of the network $n$ number of control bits in the size of the network
\tj{(however, our randomized algorithms need only logarithmic number of control bits)}.
For simplicity of analysis, we assume that every message sent during the execution
of our broadcast protocols contains the broadcast message; in practice, further optimization
of a message content could be done in order to reduce the total number of transmitted bits in real executions.
A station other than the source starts executing the broadcasting protocol
after the first successful receipt of the broadcast message; we call it
a {\em non-spontaneous wake-up model}, to distinguish from other possible settings,
not considered in this work,
where stations could be allowed to do some pre-processing
(including sending/receiving messages) prior receiving
the broadcast message for the first time.
We say that a station that received the broadcast message is {\em informed}.

\paragraph{Knowledge of stations.}
Each station knows its own ID, location, and parameters $n$, $N$.
(However, in randomized solutions, IDs can be chosen randomly from
the polynomial range such that each ID is unique with high probability.)
Some subroutines use the granularity $g$ as a parameter, though
our main algorithms can use these subroutines without being aware
of the actual granularity of the input network.
We consider two settings: one with local knowledge of density, in which each station knows also the
number of other stations in its close proximity (dependent on the $\eps$ parameter) and the other
when no extra knowledge is assumed.

\remove{
depending on the algorithm, other general network parameters such as:
diameter $D$ of the imposed communication graph,
or granularity of the network $g$, defined as the inverse of the smallest
distance between any pair of stations.\footnote{%
In many cases, the assumption about the knowledge of $D,g$ can be dropped,
by running parallel threads for different ranges of values of these parameters and implementing an additional coordination mechanism between the threads.}
}

\subsection{Grids}

%
Given a parameter $c>0$, we define 
a partition of the $2$-dimensional space
into square boxes of size $c\times c$ by the grid $G_c$, in such a way that:
all boxes are aligned with the coordinate axes,
point $(0,0)$ is a grid point,
each box includes its left side without the top
endpoint and its bottom side without the right endpoint and
does not include its right and top sides.
We say that $(i,j)$ are the coordinates
of the box with its bottom left corner located at $(c\cdot i, c\cdot j)$,
for $i,j\in \INT$. A box with coordinates
$(i,j)\in\INT^2$ is denoted $C_c(i,j)$ or $C(i,j)$ when the side of a grid
is clear from the context.

Let $\eps$ be the parameter defining the communication graph.
Then, $z=(1-\eps)r/\sqrt{2}$ is the largest value such that
the each two stations located in the same box of the grid $G_{z}$
are connected in the communication graph.
Let $\eps'=\eps/3$, $r'=(1-\eps')r=1-\eps'$ and $\gamma'=r'/\sqrt{2}$.
We call $G_{\gamma'}$ the {\em pivotal grid}, borrowing terminology from
radio networks research~\cite{DessmarkP07}.
%

Boxes $C(i,j)$ and $C'(i',j')$ are {\em adjacent} if
$|i-i'|\leq 1$ and $|j-j'|\leq 1$ (see Figure~\ref{fig:adjacent}).
For a station $v$ located in position $(x,y)$ on the plane we define its {\em grid
coordinates} $G_c(v)$ with respect to the grid $G_c$ as the pair of integers $(i,j)$ such that the point $(x,y)$ is located
in the box $C_c(i,j)$ of the grid $G_c$ (i.e., $ic\leq x< (i+1)c$ and
$jc\leq y<(j+1)c$).
\dk{The distance between two different boxes is the maximum Euclidean distance between
any two points of these boxes.
The distance between a box and itself is $0$.}

\remove{ 
A (general) {\em broadcast schedule} $\mathcal{S}$ of length $T$
wrt  $N\in\NAT$ is a mapping
from $[N]$ to binary sequences of length $T$.
A station
with identifier $v\in[N]$ {\em follows}
the schedule $\m{S}$ of length $T$ in a fixed period of time consisting of $T$ rounds,
when
$v$ transmits a message in round $t$ of that period iff
the 
position $t\mod T$ of
$\m{S}(v)$ is equal to $1$.

A {\em geometric broadcast schedule} $\mathcal{S}$ of length $T$
with parameters $N,\delta\in\NAT$, $(N,\delta)$-gbs for short, is a mapping
from $[N]\times [0,\delta-1]^2$ to binary sequences of length $T$.
Let $v\in[N]$ be a station whose grid coordinates
with respect to
the grid $G_c$ are equal to $G_c(v)=(i,j)$.
We say that $v$ {\em follows}
$(N,\delta)$-gbs $\m{S}$ 
for the grid $G_c$
in a fixed period of time, 
when $v$ transmits a message in round $t$ of that period iff
the $t$th position of
$\m{S}(v,i\mod \delta,j\mod\delta)$ is equal to $1$.
}  

\tj{
A set of stations $A$ on the plane is {\em $d$-diluted} wrt $G_c$, for $d\in\NAT\setminus\{0\}$, if
for any two stations $v_1,v_2\in A$ with grid coordinates $(i_1,j_1)$ and $(i_2,j_2)$, respectively,
the relationships $(|i_1-i_2|\mod d)=0$ and $(|j_1-j_2|\mod d)=0$ hold.}

\remove{ 
Let $\m{S}$ be a general broadcast schedule wrt $N$ of length $T$,
let $c>0$ and $\delta>0$, $\delta\in\NAT$.
A $\delta$-dilution of a  $\m{S}$ 
is defined as a $(N,\delta)$-gbs $\m{S}'$ such that the bit $(t-1)\delta^2+a\delta+b$
of $\m{S}'(v,a,b)$ is equal to $1$ iff the bit $t$ of $\m{S}(v)$
is equal to $1$. That is, each round $t$ of $\m{S}$ is
partitioned
into $\delta^2$ rounds
of $\m{S}'$, indexed by pairs $(a,b)\in [0,\delta-1]^2$, such that a station
with grid coordinates $(i,j)$ in $G_c$ is
allowed to send messages only in rounds with index $(i\mod\delta,j\mod\delta)$,
provided schedule $\m{S}$ admits a transmission in its (original) round $t$.
%
%
} 

\begin{figure}
\begin{center}
\epsfig{file=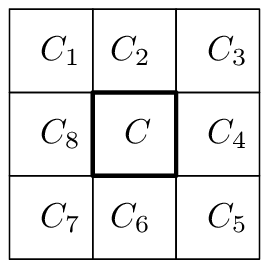, scale=0.8}
\end{center}
\vspace*{-2ex}
\caption{%
The boxes
$C_1,\ldots,C_8$ are adjacent to $C$.}
\label{fig:adjacent}
\vspace*{-2ex}
\end{figure}%
\remove{ 
Observe that, since ranges of stations are equal to the length
of diagonal of boxes of the pivotal grid, a box $C(i,j)$ can have at most
$20$ neighbors (see Figure~\ref{fig:adjacent}).
We define the set $\DIR\subset[-2,2]^2$ such  that $(d_1,d_2)\in\DIR$ iff
it is possible that boxes with coordinates $(i,j)$ and $(i+d_1,j+d_2)$
can be neighbors.
Given $(i,j)\in\INT^2$ and $(d_1,d_2)\in\DIR$, we say that the box $C(i+d_1,j+d_2)$
is {\em located in direction} $(d_1,d_2)$ from the box $C(i,j)$.
} 

\comment{%
\tj{For a family $\m{F}=(F_0,\ldots,F_{k-1})$ of subsets of $[N]$, an {\em execution}
of $\m{F}$ on a set of stations $V$ is a protocol in which $v\in V$ transmits in
rund $i$ iff $v\in F_{i\mod k}$.
}
A family $S$ of subsets of $[N]$ is a {\em $(N,k)$-ssf (strongly-selective family)}
if, for every non empty subset $Z$ of $[N]$ such that
$|Z|\leq k$ and for every element $z\in Z$, there is a set $S_i$ in $S$ such that
$S_i\cap Z=\{z\}$.
It is known that there exists $(N,k)$-ssf of size $O(k^2\log N)$ for every $k\leq N$,
c.f.,~\cite{ClementiMS01}.
Let $k=(2d+1)^2$, let $S$ be a $(N,k)$-ssf,
$s=|S|=O(\log N)$,
$N=poly(n)$.
The sets $S_0,\ldots,S_{s-1}$ of the family $S$ define a broadcast schedule $S'$ in such a way
that station $v$ transmits in round $t$ iff $v\in S_{t\mod s}$. (That is, $S'_{t\mod s}[v]=1$ iff
$v\in S_{t\mod s}$.)
}

\section{Randomized Algorithms}
\label{s:random}

We first present and analyze an algorithm relying on the knowledge of local density,
and then we proceed with general protocol that does not require such knowledge in the input.



\subsection{Algorithms for Known Local Density}
\labell{s:dknown}
\input{RandKnown.tex}

\section{Experimental Results}

\input{experiments.tex}

\section{Algorithms for Unknown Local Density}
\labell{s:dunknown}

\input{RandUnknown.tex}



\section{Conclusions and Future Work}

In this work we showed the first provably well-scalable distributed solutions for the broadcast problem
in {\em any} wireless networks under the SINR physical model.
Additionally, one of our algorithms compare favorably
with the classical heuristic based on exponential backoff protocol,
which we demonstrated by simulations on uniform and social networks.
The other algorithms, also provably well-scalable,
provide several novel techniques for leader election and broadcast,
which may be adopted for the purpose of other communication problems.


\bibliographystyle{abbrv}

\input{bibliography-new.bbl}


\clearpage

\appendix

\input{RandProofs.tex}

\end{document}

%% file: RandKnown.tex
In this section we describe our broadcasting algorithm
for networks of known local density. That is every station $v$
knows the total number of stations $\Delta=\Delta(v)$ in its box of the grid $G_\gamma$.
In this section $\gamma = \frac{\eps}{2\sqrt{2}}$.
The global value ``counter'' is transmitted by nodes together with the message.
We say that $(i,j)\equiv (a,b)\mod d$ if and only if $i\equiv a \mod d$ and $j\equiv b\mod d$.

\begin{algorithm}[H]
	\caption{RandBroadcast$(\Delta,d,T)$ \Comment{code for node $v$}} 
	\label{alg:density}
	\begin{algorithmic}[1]

    \If{$v$ is the source} $v$ transmits \EndIf
    \For{$counter=1,2,3,\ldots,T$}
        \For{each $a,b:0\le a,b<d$}
            \If{$v \in C(i,j): (i,j) \equiv (a,b) \mod d$}
                \State $v$ transmits with probability $1/\Delta$
            \EndIf
         \EndFor
    \EndFor
    \end{algorithmic}
\end{algorithm}

\paragraph{Analysis of time performance of RandBroadcast in any network.}
We start with proving three general properties regarding interference in the SINR model.

\begin{fact}
\labell{fact01}
If the interference at the receiver is at most $\cN\alpha x$, then
it hears the transmitter from the distance $1-x$.
\end{fact}

\def\FACTONE{
\noindent
{\bf\em Proof of Fact~\ref{fact01}:}
By the Bernoulli inequality we get
$(1+x)^\alpha \ge 1+\alpha x$.
Thus
\[
SINR\ge
\frac{P}{(\cN+\cN\alpha x)(1-x)^\alpha}
\ge
\frac{P}{\cN (1+x)^\alpha(1-x)^\alpha}
\]
\[
=
\frac{P}{\cN (1-x^2)^\alpha}
\ge 
\frac{P}{\cN}=\beta.
\]
\qed
}

We say that a function $d_{\alpha}:\NAT\to\NAT$ is {\em flat} for $\alpha\geq 2$ if
\begin{equation}
d_{\alpha}(n)=\left\{
\begin{array}{rcl}
O(1) & \mbox{ for } & \alpha>2\\
O(\log n) & \mbox{ for } & \alpha=2
\end{array}
\right.
\end{equation}

Consider the following process on the grid $G_\gamma$ for some $d\in\mathbb{N}_+$.
For every box $C(i\cdot d,j\cdot d)$, where $(i,j)\neq (0,0)$,
a number $x(i,j)$ is chosen at random, 
according to some probability distribution with
expected value of at most~$1$.
Next, in every box $C(i\cdot d,j\cdot d)$, where $(i,j)\neq (0,0)$,
$x(i,j)$ locations for transmitting stations are selected by an adversary.
These stations cause some interference in stations in boxes $C(i',j')$ of distance
at most $1$ from the box $C(0,0)$.
We denote the expected value of the maximum of these interferences by $I_d$,
where the maximum is taken over all possible locations of stations selected by the adversary
and over all possible locations of receivers in boxes $C(i',j')$ of distance
at most $1$ from the box $C(0,0)$.

Let $s_\alpha=\min\left\{\frac{\ln n}{2}+\ln 2,\frac{1}{2^{\alpha-2}(\alpha-2)}\right\}+\frac{1}{2^{\alpha}(\alpha-1)}$ and 
$d_{\alpha,I,\gamma}=\left\lceil\frac{1}{\gamma} \left(\frac{8Ps_\alpha}{I}\right)^{1/\alpha}\right\rceil$.

\begin{lemma}
\labell{lemma02}
Consider the process described above.
Then, for any $I>0$ there exists a flat function $d=d(n)$
such that $I\ge I_d$.
Moreover, 
for $I\le \frac{8Ps_\alpha}{2^\alpha}$ we have $I_d\le I$
when 
$d=d_{\alpha,I,\gamma}$.
\end{lemma}

\def\LEMMATWO{
\noindent
{\bf\em Proof of Lemma~\ref{lemma02}:}
Let us estimate the expected maximum interference $I_d$ for an arbitrary $d$.
By this maximum we mean the biggest interference over all points
of the boxes in the distance at most 1 from $C(0,0)$ when the transmitting stations are put to
the boxes $C(i\cdot d,j\cdot d)$ by the adversary.
\[
I_d\le \sum_{k=1}^{\sqrt n/4} \frac{8kP}{(kd\gamma-1)^\alpha}.
\]
If we denote $t=d\gamma$ and assume, that $t>2$ we get
\[
I_d\le \sum_{k=1}^{\sqrt{n}/4} \frac{8kP}{((k-1/2)t)^\alpha}
   \le \frac{8P}{t^\alpha} \sum_{k=1}^{\sqrt{n}/4} \frac{k}{(k-1/2)^{\alpha}}
\]
and
\[
I_d \le \frac{8P}{t^\alpha}\int_1^{\sqrt{n}} \frac{x}{(x-1/2)^{\alpha-1}}.
\]
When $\alpha>2$ we can estimate this expression as follows
\[
I_d\le \frac{8P}{t^\alpha}\int_1^\infty \frac{x}{(x-1/2)^{\alpha-1}}
\]
so
\[
I_d\le \frac{8P}{t^\alpha}\left(\frac{1}{2^{\alpha-2}(\alpha-2)}+\frac{1}{2^{\alpha}(\alpha-1)} \right).
\]
There is also another way to estimate $I_d$, that works also if $\alpha=2$
\[
I_d\le \frac{8P}{t^\alpha}\left(\int_1^{\sqrt{n}} \frac{1}{x-1/2}+\int_1^{\infty} \frac{1}{2(x-1/2)^\alpha}\right)
\]
and
\[
I_d\le \frac{8P}{t^\alpha}\left(\frac{\ln n}{2}+\ln 2+\frac{1}{2^{\alpha}(\alpha-1)} \right).
\]
Thus
\(\displaystyle
I_d\le \frac{8P}{t^\alpha} s_\alpha.
\)

\noindent
Since $t>2$ we always have 
\(\displaystyle
I_d\le \frac{8P}{2^\alpha} s_\alpha.
\)

\noindent
If we want to get $I_d\le I$ we should take such $d=d_{\alpha,I,\gamma}$, that
\[
\frac{8P}{(d\gamma)^\alpha} s_\alpha\le I.
\]

\noindent
So we can good choose $\displaystyle d_{\alpha,I,\gamma}=\left\lceil\frac{1}{\gamma} \left(\frac{8Ps_\alpha}{I}\right)^{1/\alpha}\right\rceil$.
\qed
}

\begin{corollary}\labell{corollary03}
If in the above described process, the expected number of stations in a box is $x$ instead of 1,
then for any $d$ we have the maximum expected interference in boxes in the distance at most $1$ 
equal to $x\cdot I_d$, where $I_d$ is as described in Lemma~\ref{lemma02}.
\end{corollary}

%

We proceed with the analysis of algorithm RandBroadcast.


\begin{fact}\labell{fact11}
Consider a  round of algorithm RandBroadcast, different from the first one.
The probability that in a box $(i,j) \equiv (a,b)$
exactly one station transmits is bigger than~$1/e$.
\end{fact}

\def\FACTELEVEN{%
\noindent
{\bf\em Proof of Fact~\ref{fact11}:}
Because of the inequality $1/e<(1-1/n)^{n-1}$ we get
\[\Pr(\text{exactly on node transmits})=
  \Delta\frac{1}{\Delta}\left(1-\frac{1}{\Delta}\right)^{\Delta-1}>\frac{1}{e}
\ .
\]
\qed
}


\begin{fact}\labell{fact12}
Consider a  round of algorithm RandBroadcast$(\Delta,d,T)$ for
 $d=d_{\alpha,{\cal N}\alpha\eps/4,\gamma}$, different from the first one.
The probability that exactly one station
in box $C(i,j)$, where $(i,j) \equiv (a,b)$, transmits
and the interference from other stations measured in all boxes connected
with box $C(i,j)$ is smaller or equal to ${\cal N}\alpha\eps/2$
is bigger than $\frac{1}{2e}$.
\end{fact}

\def\FACTTWELVE{%
\noindent
{\bf\em Proof of Fact~\ref{fact12}:}
By Markov inequality and Lemma \ref{lemma02} we can bound the probability
that the maximum interference from
boxes different than $C(i,j)$ measured in boxes in distance at most $1$ to $C(i,j)$
exceeds ${\cal N}\alpha\eps/2$.
We take advantage of the equality $d=d_{\alpha,{\cal N}\alpha\eps/4,\gamma}$.
\[
P_1\le \frac{E(\text{max interference})}{{\cal N}\alpha\eps/2}\le \frac{1}{2}
\ .
\]
Thus, by the Fact~\ref{fact11}, the probability $P_2$ that exactly one node in $C(i,j)$ transmits 
and is heard in the distance at most $1-\eps/2$ is bounded as follows
\[
P_2\ge\frac{1}{e}(1-P_1)\ge \frac{1}{2e}
\ .
\]  
}


\begin{lemma}
\labell{lemma13}
Consider a Bernoulli scheme with success probability $p<1-\ln 2$.
The probability of obtaining at most $D$ successes in
$2D/p + 2\ln(1/\delta) / p$ trials is smaller than $(D+1)\delta$.
\end{lemma}

\def\LEMMATHIRTEEN{%
\noindent
{\bf\em Proof of Lemma~\ref{lemma13}:}
Let $t=2D/p + 2\ln(1/\delta) / p$.
If $0\le i\le D$, then by the inequality $1+x\le e^x$
\begin{eqnarray*}
\lefteqn{\Pr(\text{exactly $i$ successes})= {t \choose i} p^i (1-p)^{t-i}<}\\
     &&{t \choose D} p^D (1-p)^{t-D}<t^D e^D D^{-D} p^D e^{pD-pt}
\ .
\end{eqnarray*}
So
\[
\Pr(\text{at most $D$ successes})=
    (D+1) t^D e^D D^{-D} p^D e^{pD-pt}
\ .
\]
Since $p<1-\ln 2$, then $2^D e^{pD} e^{-D}<1$.
Let  $pt=2D+x$.
We have
{\setlength\arraycolsep{0.1em}
\begin{eqnarray*}
\Pr(\text{$\le D$ successes}) &< &
    (D+1) (2D+x)^D e^D D^{-D} e^{pD-2D-x}\\
     &=&
    (D+1) \left(1+\frac{x}{2D}\right)^D 2^D e^{pD} e^{-D} e^{-x}\\
    &<&
    (D+1) e^{x/2} e^{-x} = (D+1) e^{-x/2}
\ .
\end{eqnarray*}}
When $x=2\ln(1/\delta) / p$ we have
\[
\Pr(\text{at most $D$ successes})< (D+1)\delta
\ .
\]
\qed
}

We say that a subset of nodes $W$ of graph $G$ is an {\em $l$-net} if 
any other node in $G$ is in distance at most $l$
from the closest node in $W$.


\begin{fact}\labell{fact14}
If $G$ is of eccentricity $D$, then there exists a $(1-\eps)$-net
$W$ of cardinality at most $4(D+1)^2$.
\end{fact}

\def\FACTFOURTEEN{%
\noindent
{\bf\em Proof of Fact~\ref{fact14}:}
Let $q=1-\eps$.
Ranges $q$ of all the stations must be
all inside the circle of radius $(D+1)q$.
The area of this circle is $\pi(D+1)^2 q^2$.
Let us greedily pick a maximal set of nodes
such that any two nodes are in distance at
least $q$. This set is a $q$-net $W$.
Let us estimate the cardinality of $W$.
All the circles of radius $q/2$ and center
belonging to $W$ are disjoint and have areas $\pi q^2$.
They have total area at most  $\pi(D+1)^2 q^2$,
so $|W|\le 4(D+1)^2$.
\qed
}

Using the above results we conclude the analysis.

\begin{theorem}
\labell{theorem15}
Algorithm RandBroadcast$(\Delta,d,T)$ completes broadcast in any $n$-node network in time
$O(d^2(D+\log(1/\delta)))$ with probability $1-\delta$,
for $d=d_{\alpha,{\cal N}\alpha\eps/4,\gamma}$ and some $T=O(D+\log(1/\delta))$.
\end{theorem}

\begin{proof}
To complete broadcasting it is enough that
all the boxes containing stations of the $(1-\eps)$-net $W$
transmit the message at least once
and are heard by all their neighbors.
Such a box containing $v\in W$ transmits successfully if
the message is successfully transmitted at most $D$ times
on the shortest path from the source to $v$ in $G$, and
finally is successfully transmitted by the box containing $v$.
The sufficient condition for this to happen is that a chain of
altogether at most $D+1$ successful transmissions heard
by all potential receivers occurs.
The probability of a successful transmission within this chain is 
bigger than $p=\frac{1}{2e}$, by Fact \ref{fact12}
(recall that Fact \ref{fact12} uses our assumption $d=d_{\alpha,{\cal N}\alpha\eps/4,\gamma}$).

Now we estimate the probability that algorithm RandBroadcast
completes the broadcast. Let the number of trials be
$T=2D/p + 2\ln(1/\delta') / p$, for some $\delta'\in\mathbb{R}$.
By Lemma \ref{lemma13}, Fact \ref{fact01} and Fact \ref{fact14}

\vspace*{-2ex}
\[
\Pr(\text{All $v\in W$ transmit successfully})
\ge
\]

\vspace*{-5ex}
\[
1-\sum_{v\in W} \Pr(\text{$v$ doesn't transmit successfully})
\ge
1-4(D+1)^3\delta'
\ .
\]

\vspace*{-1ex}
\noindent
This is bigger than $1-\delta$ for our choice of $T$.
Note also that $T=O(D+\log(1/\delta))$.
Because we have a trial every $d^2$ rounds,
we need altogether $O(d^2(D+\log(1/\delta)))$ rounds,
for $d=d_{\alpha,{\cal N}\alpha\eps/4,\gamma}$.
\end{proof}

%% file: experiments.tex



In the experiments we used two kinds of randomly generated networks: 
uniform and social, see Figure~\ref{fig:nets}
for examples. Each network was guaranteed to be strongly connected
(i.e., non-connected networks were removed).

Uniform networks were generated by adding random nodes
--- uniformly distributed on $S \times S$ size square --- 
until desired size $n$ was achieved.

Our social networks are generated in a way which accustoms modeling of graph-based social networks 
to geometrical constraints.
We divided the surface $S \times S$ into a grid of size $\epsilon \times \epsilon$. Each box $t_i$ in the grid was assigned a weight $w_i = |\{ v : v \mbox{ is in the distance of $2$ from any node in } t_i\}|$.
Before adding a new node to a network, we chose between two modes of addition.
With probability $p=0.9$ the first mode was chosen. 
Within this mode a box $t_i$ was chosen with probability proportional to its weight, 
and then a new node was located randomly in the selected box according to a uniform distribution. With probability $1-p=0.1$ we used the second mode, in which a new node on a random position within the square $S \times S$ is located. After adding a new node we update weights of boxes.
Nodes are added until the size $n$ is achieved. 

\begin{figure}[th!]
\begin{center}
\vspace*{-2ex}
\epsfig{file=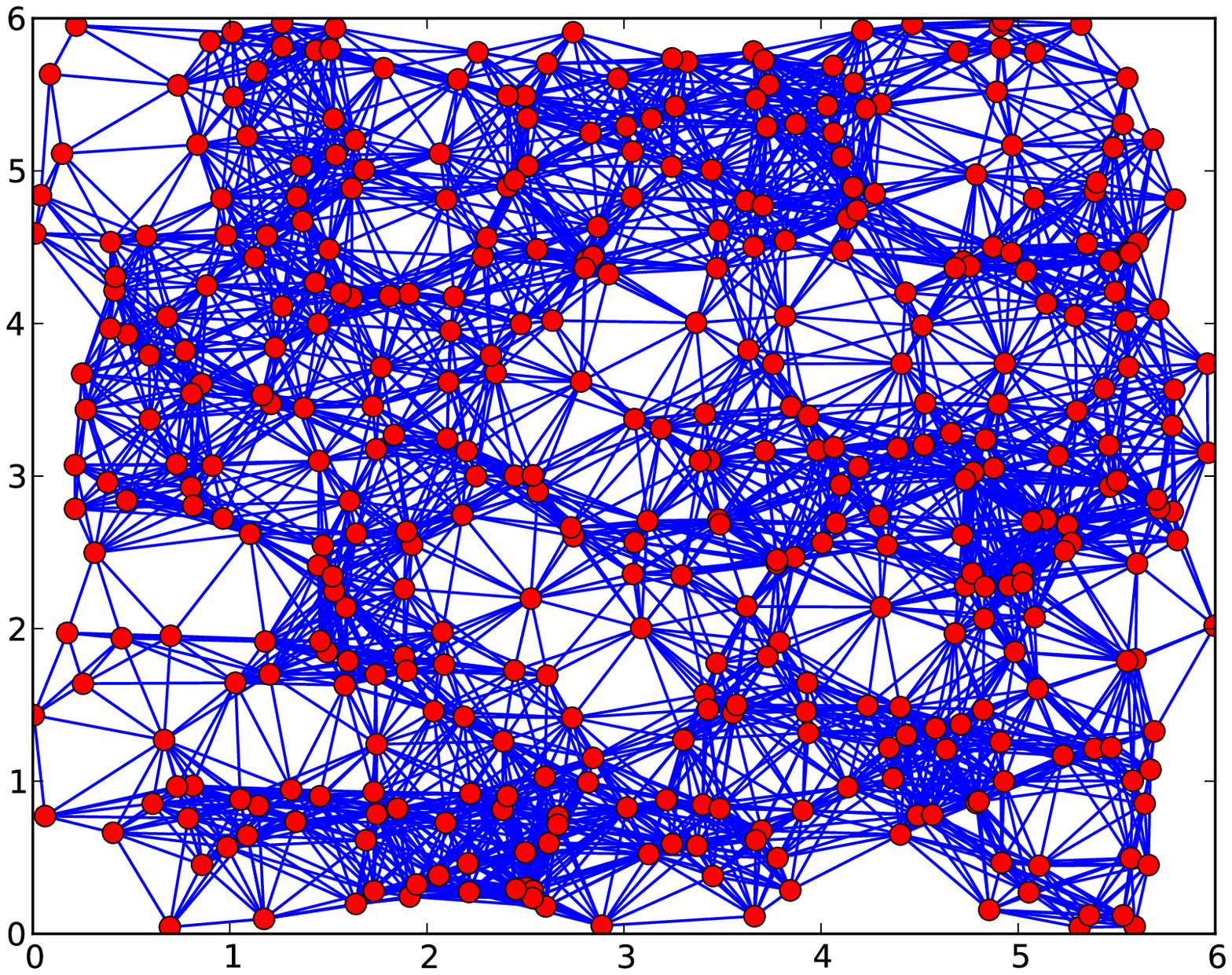, scale=0.44}
\\ \vspace{-2ex}
\epsfig{file=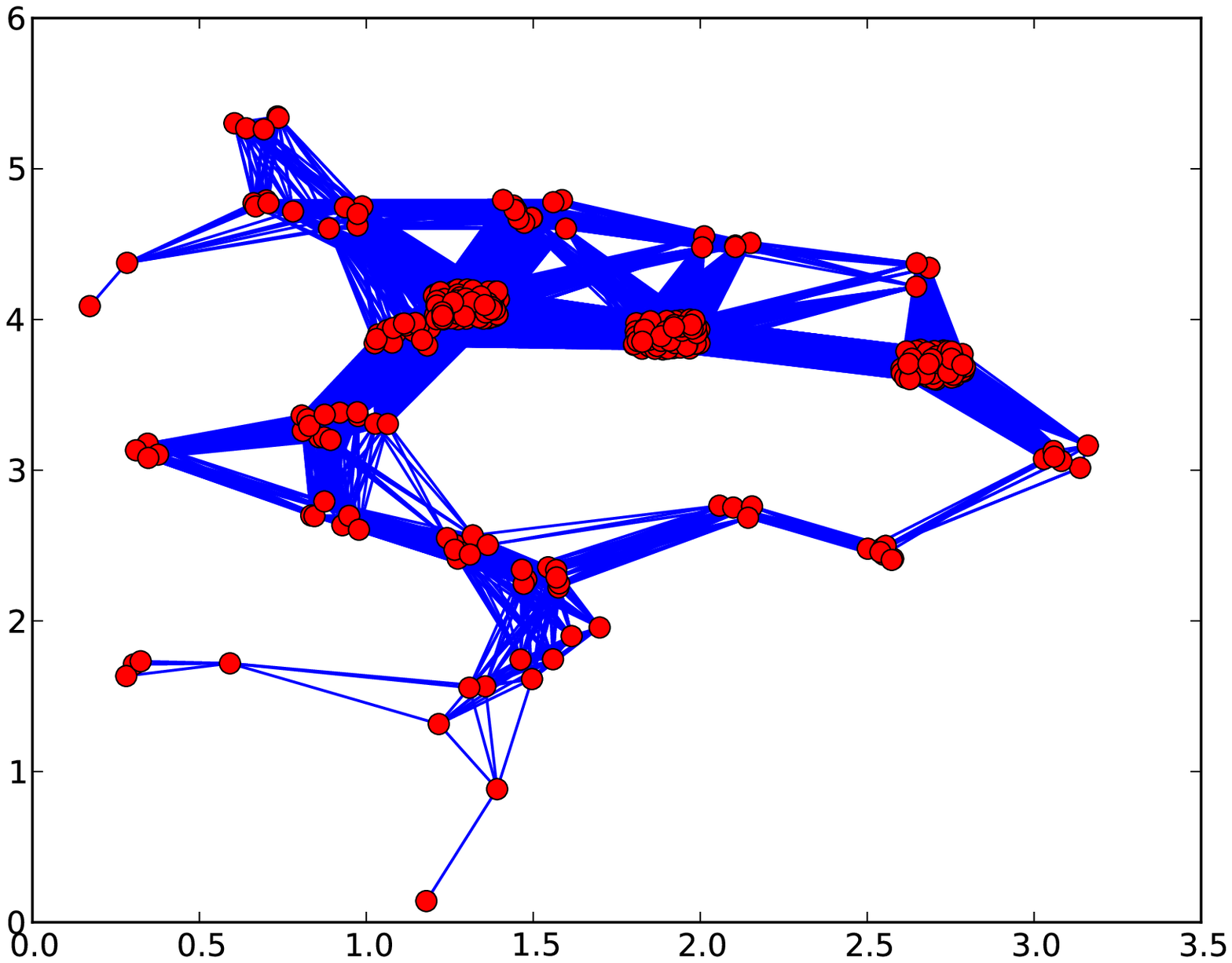, scale=0.44}
\end{center}
\vspace*{-4ex}
\caption{Examples of uniform (top) and social (bottom) networks with $n=400$ nodes, distributed
in a $6\times 6$ square.}
\label{fig:nets}
\vspace*{-2ex}
\end{figure}%

We tested the performance of algorithm RandBroadcast (Algorithm~\ref{alg:gran}) 
with parameter $d=10$
and compared it with the exponential backoff protocol.
To neutralize possible advantage of Algorithm~\ref{alg:gran} coming from the knowledge of
local density $\Delta$, 
we allowed backoff algorithm to use this knowledge as well to limit the number of iterations.
%
More precisely, each node that received the broadcast message, 
transmits the message in a random round of consecutive time periods ({\em windows}) of sizes: 
$2^0, 2^1, \ldots, 2^{\log\Delta}$. 
If the node receives acknowledgment of its message, 
it starts again with window of size $2^0$. 
If no new acknowledgment message is received in a sequence of windows of sizes 
$2^0,\ldots,2^{\log\Delta}$, the node terminates its execution of backoff protocol.

In our experiments we set the following set of parameters of the SINR model: 
$\alpha = 2.5$, $\mathcal{N} = 1$, $\beta = 1$, $\epsilon = 0.2$.


\tj{For each $n\in\{50,100,150,200, 400, 600, 800, 1000,1500,$ $2000\}$ and $S=6$, 
we generated $20$ networks with $n$ nodes located on a square of size $S\times S$.} 
Then, both algorithms were executed on each network. For each $n$,
we calculated average time over all $20$ networks generated with these parameters. 
Moreover, in order to check scalability of Algorithm~\ref{alg:gran} 
(whose asymptotic time complexity is proportional to $D$, the eccentricity of the source), 
we also present graphs illustrating average proportion of time complexity
and $D$. 
The results for uniform networks are presented on Figure~\ref{fig:uniform},
and the results for social networks are given on Figure~\ref{fig:social}.
The main conclusion is that exponential backoff protocol --- although very efficient
for networks with relatively small number of users (roughly below $600$),
is not scalable, while the average time performance of RandBroadcast is
away from the absolute lower bound $D$ by a small constant for uniform
and social networks of at least $1000$ nodes.

\begin{figure*}
\vspace*{-2ex}
\begin{center}
\epsfig{file=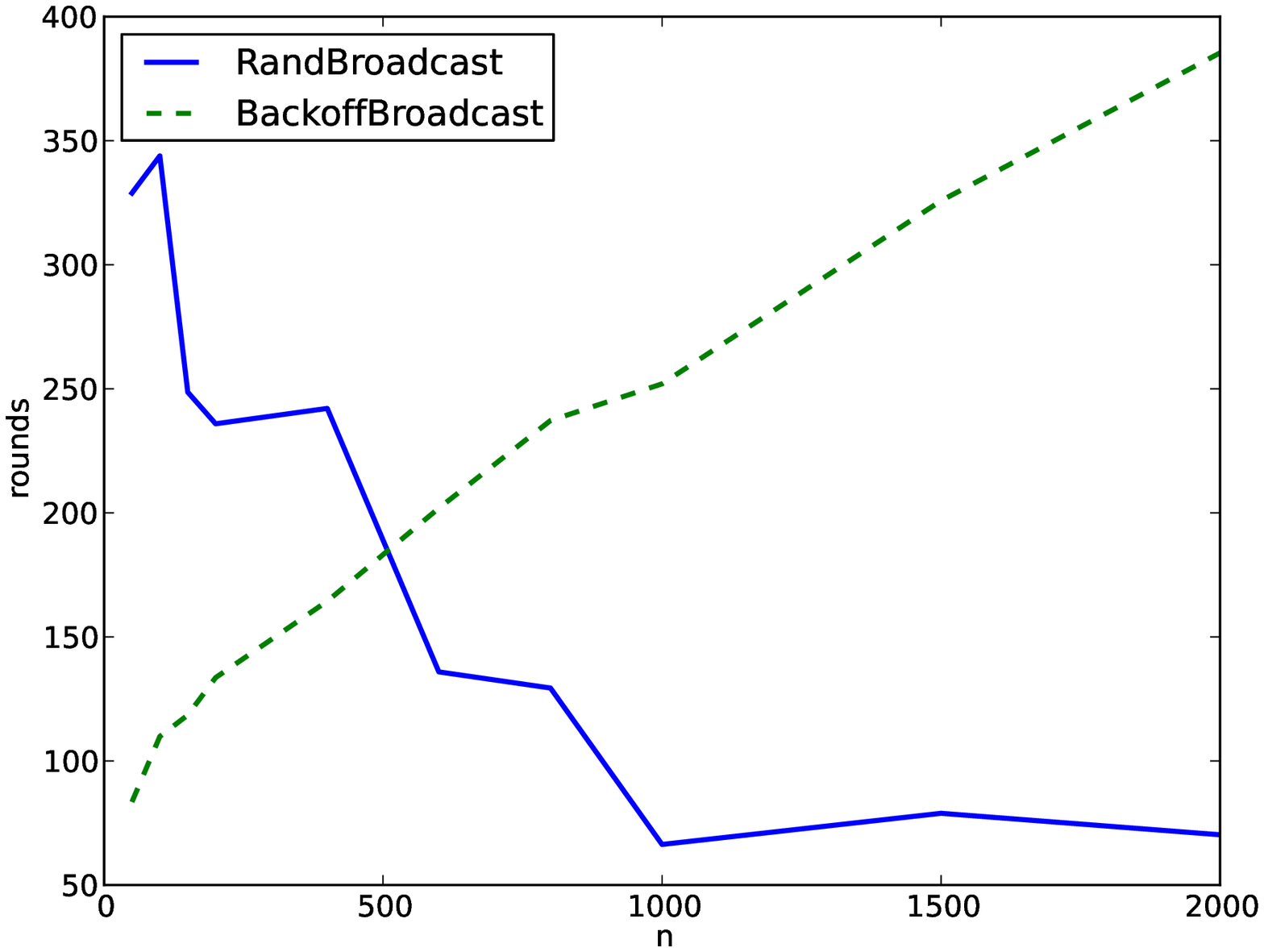, scale=0.43}
\epsfig{file=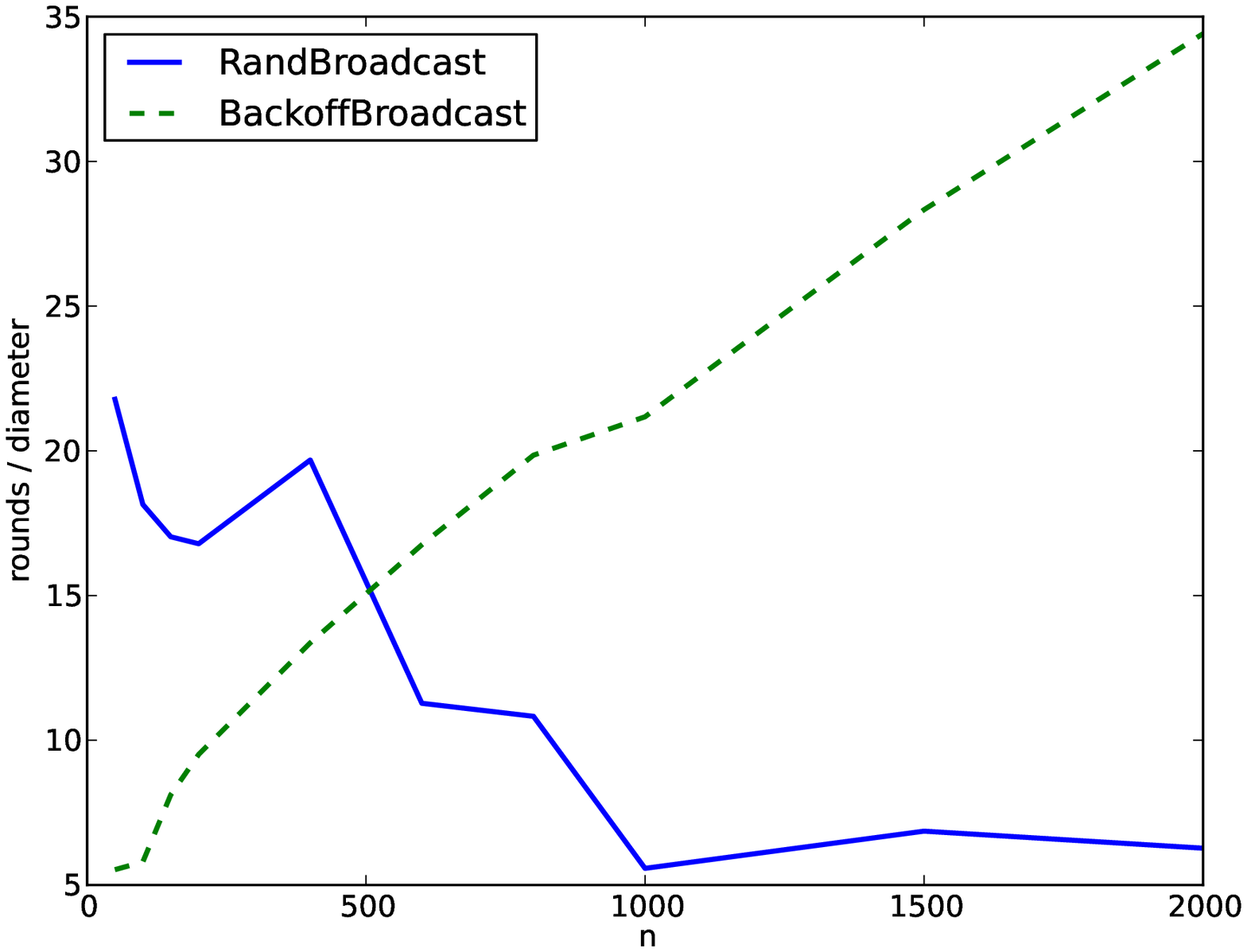, scale=0.43}
\end{center}
\vspace*{-4ex}
\caption{Simulation results for uniform networks: average time (left) and the ratio of time over diameter (right).}
\label{fig:uniform}
\end{figure*}

\begin{figure*}
\begin{center}
\vspace*{-3ex}
\epsfig{file=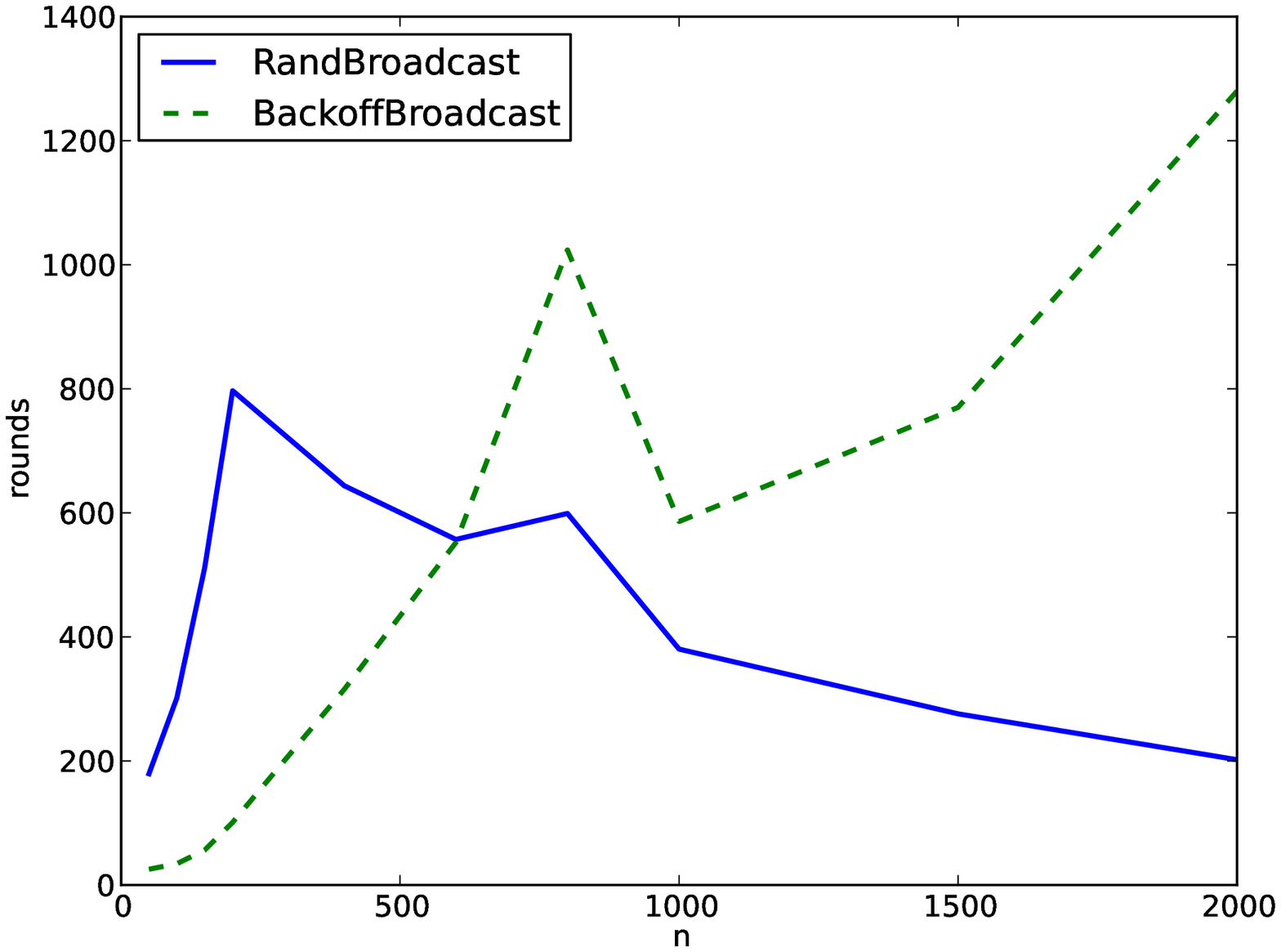, scale=0.43}
\epsfig{file=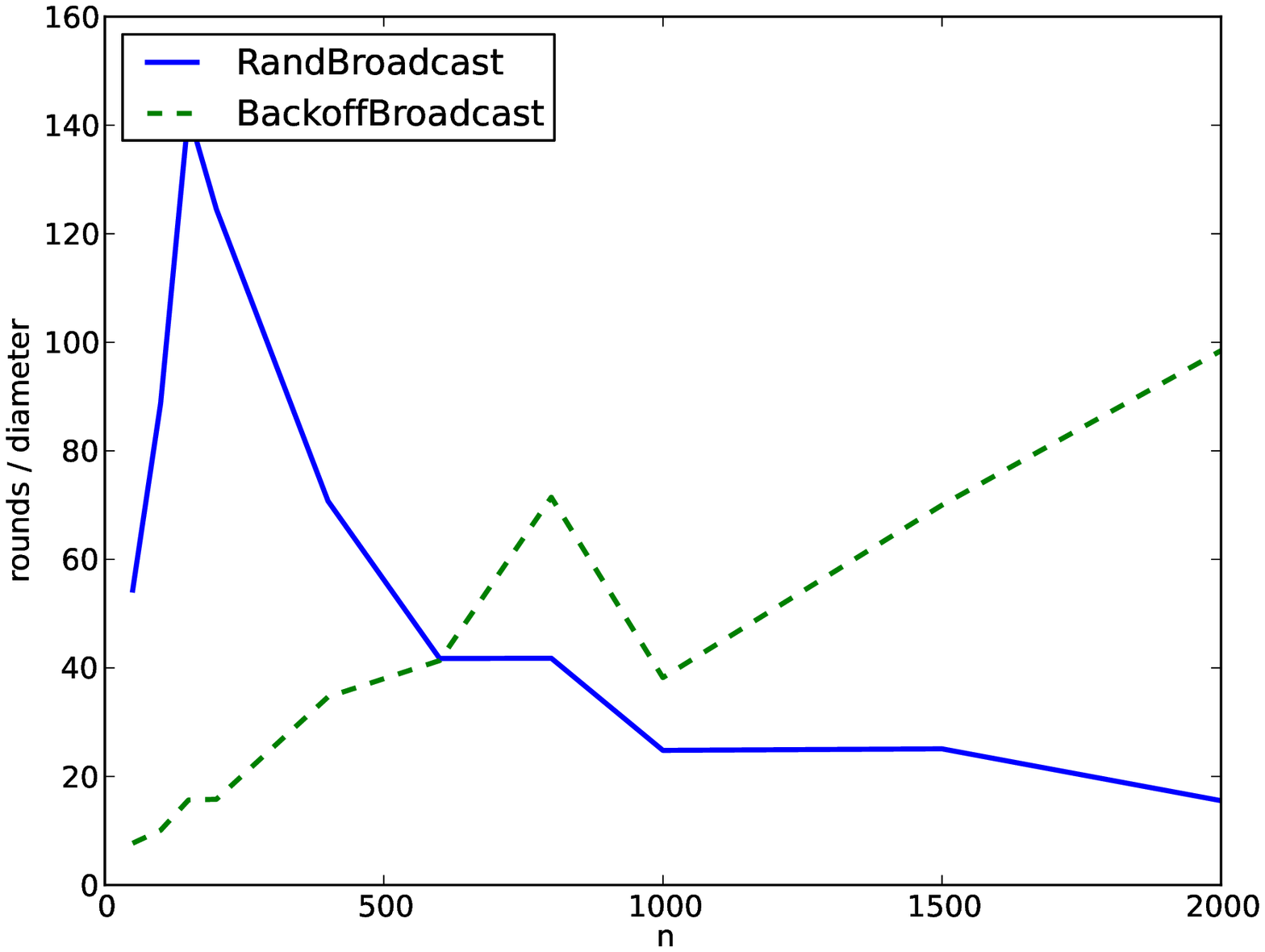, scale=0.43}
\vspace*{-4ex}
\end{center}
\caption{Simulation results for social networks: average time (left) and the ratio of time over diameter (right).}
\label{fig:social}
\vspace*{-3ex}
\end{figure*}


%% file: RandUnknown.tex
In this section we describe our broadcasting algorithm
for networks of unknown local density.
To~construct this algorithm we consider the grid $G_\gamma$, where
$\gamma = \frac{\eps}{6\sqrt{2}}$.
Due to Fact~\ref{fact01} if the interference at the receiver does not exceed ${\cal N}\alpha\eps/2$,
then a node can hear the transmitter in the distance $1-\eps/2$.
Assume now, that we have two boxes $V$ and $U=C(i,j)$ and nodes $v,u$ such that
$v\in V,u\in U$ and $\{u,v\}\in G$.
In such a setting if a single node from $V$ transmits the message,
then it is heard by all stations in all boxes $C(i+a,j+b)$, where $a,b\in [-2,2]$.

For this section we modify the notion of boxes being adjacent.
Two boxes $V$ and $U$ are adjacent if the distance between them is at most $1-\eps/2$.
But with one exception -- boxes that are very close each to other are not adjacent.
More precisely the box $C(i,j)$ is not adjacent to
any box $C(i+a,j+b)$, where $a,b\in [-2,2]$.
Whenever we have two boxes $V=C(i_V,j_V)$ and $U=C(i_U,j_U)$ and nodes $u,v$ such that
$v\in V,u\in U$ and $\{u,v\}\in G$,
the box $V$ is adjacent to all boxes $C(i_U+a,j_U+b)$ where $a,b\in [-2,2]$
unless $C(i_U+a,j_U+b)=C(i_V+a',j_V+b')$ where $a',b'\in [-2,2]$.
In other words if $\{u,v\}\in G$, then $V$ is adjacent to all boxes that are too
close to $U$ to be adjacent to $U$, unless these boxes are also too close to $V$.

The {\em neighborhood} of a box $V$ is the set of all boxes $U$ adjacent to $V$.
To formulate the algorithm we have to define the {\em octant} of the
neighborhood of the box $V=C_{i,j}$.
In order to do it we place on the plane a Cartesian coordinate system
with the origin in the center of the box $V$.
This coordinate system is naturally subdivided into four {\em quadrants}
i.e.\ the plane areas bounded by two reference axes forming the $90\degree$ angle.
The quadrant can be divided by the bisector of this angle into two {\em octants}
corresponding to the angle of $45\degree$.
We attribute one of the rays forming the boundaries of the octants
to each octant, so that they are disjoint (and connected) as the subsets of the plane.
An {\em octant} of the neighborhood of $V$ is the set of all boxes $U$
in the neighborhood that have centers in a given octant of the coordinate system.

\begin{fact}\labell{fact21}
In the octant of the neighborhood of $V$ each two stations
are in the distance at most $(1-\eps/2)$.
\end{fact}


\begin{algorithm*}[t!]
	\caption{RandUnknownBroadcast$(d,T)$}
	\label{alg:unknown}
	\begin{algorithmic}[1]
    \State the source $s$ transmits and becomes the leader of its box of $G_{\gamma}$
    \For{counter$\gets 1,2,\ldots,T$}
        \For{each $a,b:0\le a,b<d$}
            \If{$v$ is the leader of $V=C(i,j):(i,j) \equiv (a,b) \mod d$}
                $v$ transmits
            \EndIf
        \EndFor
        \For{each $a,b:0\le a,b<\bar{d}$}
            \For{each octant of the neighborhood of $V=C(i,j):(i,j) \equiv (a,b) \mod \bar{d}$}
                \State $U\gets$ box in the octant with a leader of lexicographically smallest coordinates
                \State $u\gets$ leader od $U$
                \State conflict$(v)\gets$ false
                \For{$k=0,1,2,3,...,\log n$}
                    \While{$U$ exists and $V$ has no leader and not conflict$(v)$}
                        \State K1: Each vertex $v\in V$ transmits with the probability $(1/n)2^k$
                        \State K2: {\bf if} $u$ hears $v$ in K1 {\bf then} $u$ transmits ``$v$'' and $v$ becomes the leader
                        \State \phantom{K2: }{\bf if} $v$ transmitted in K1 and hears nothing in K2 {\bf then} conflict$(v)\gets$ true
                        \State K3: nodes $v$ transmitting in K1 and $u$ transmit
                        \State \phantom{K3: }{\bf if} $v$ not transmitting in K1 does not hear $u$ {\bf then} conflict$(v)\gets$ true
                    \EndWhile
                \EndFor
            \EndFor
        \EndFor
    \EndFor
    \end{algorithmic}
\end{algorithm*}

We should add a couple of words of explanation to our algorithm.
In the algorithm $d=d_{\alpha,{\cal N}\alpha\eps/2,\gamma}$ and 
$\bar{d}=\left\lfloor 1/\gamma\right\rfloor d_{\alpha,{\cal N}\alpha\eps/28,\gamma\lfloor 1/\gamma\rfloor}$.
The algorithm consists of $T$ iterations of the most external loop.
Each of these iterations consists of two parts.
The first part is a deterministic broadcast from the leaders of the boxes
to all nodes in the distance at most $1-\eps/2$ from these leaders.
It is assumed that new vertices are woken up only in the very beginning and
in the first part.
The second part is a probabilistic algorithm attempting to elect the leaders
in all the boxes in which the message was heard in the first part
and which currently do not have leaders.
To make such an attempt in the box $V$ some help
from the leader of a box $U$ adjacent to $V$ is needed.
This attempt is made separately for each octant.
Within an octant the leaders hear each other in the first part,
so they all can say without any additional communication
which of them has lexicographically smallest coordinates.
Also any vertex in $V$ knows whether any leader in an octant exists.

The loop ``{\bf for} $k$'' assures that in round K1 the transmission
probability grows twice per iteration starting from $1/n$.
Rounds K2 and K3 are designed so that they ``switch off''
till the end of the loop ``{\bf for} $k$''
all nodes of $V$ when any of them transmits in K1.
This assures, that the expected interference caused by the computation in $V$ is small.
There are three possible outcomes of the round K1.
One of them is that no vertex transmits in K1.
In such a case all nodes in $V$ hear $u$ in K3 and $k$ is incremented
unless the interference jams $u$ in K3 and nodes in $V$ switch off.
The next possibility is that exactly one node $v\in V$ transmits on K1.
In such a case all nodes in $V$ are notified that $v$ is the leader in K2
unless the interference jams $v$ in K1 and nodes in $V$ switch off in K2 and K3.
Note that no vertex of $V$ can hear $u$ in K3 because $v$ is closer to this
vertex than $u$.
The last possibility is that more than one node of $V$ transmits on K1.
We have two subcases.
The first subcase is that one of these nodes in $V$
is heard by $u$ (can happen for some $\beta,{\cal N},P,\alpha$) and this vertex becomes elected in K2.
The second subcase is that $u$ hears nothing in K1.
Again no vertex of $V$ can hear $u$ in K3 because the transmitting nodes in $V$
are closer to this vertex than $u$.

Now we prove an analog of Fact~\ref{fact14} for our algorithm.

\begin{lemma}\labell{lemma22}
Let $G$ be of the eccentricity $D$.
There exists a set of boxes $W$ of the grid $G_\gamma$
of cardinality at most $4(D+1)^2$ having the two following properties
\begin{itemize}
\item if we choose one station from each box then these stations form
      a $(1-\eps/2)$-net in the set of all the stations,
\item for each box of $W$ there exists a sequence of at most $D+1$ boxes
      beginning from the box containing the source
      and ending in this box such that two consecutive boxes are adjacent.
\end{itemize}
\end{lemma}

\def\LEMMATWENTYTWO{
\noindent
{\bf\em Proof of Lemma~\ref{lemma22}:}
In Fact~\ref{fact14} we proved, that in $G$ exists a $(1-\eps)$-net $W'$
of cardinality at most $4(D+1)^2$.
Obviously there exist paths from the source to each node of $W'$
of the lengths at most $D$.
We now say how one can obtain $W$ having $W'$.
Let us consider one of the nodes $v$ belonging to $W'$.
There is a path $v_0,v_1,v_2,\ldots,v_d=v$ where
$v_0$ is the source and  $d\le D$.
We assume this is the shortest path from $v_0$ to $v$.

We can replace each vertex $v_i$ of the path by its box $V_i$
obtaining a sequence of at most $D+1$ boxes.
Any vertex of $V_d$ is in the distance at most $1-\eps/2$ from
any vertex in the distance at most $1-\eps$ from $v$
(it can replace $v$ in $W'$ as a $(1-\eps/2$)-net).
We show how we can modify the sequence $V_i$
so that all subsequent boxes are adjacent.
Assume some are not adjacent,
which can happen when they are too close each to other.
That is $V_k=C(i,j)$ and $V_{k+1}=C(i+a,j+b)$ where $a,b\in [-2,2]$.
In such a case we remove the box $V_{k+1}$ from the sequence.
If $v\not\in V_{k+1}$, then the boxes $V_k$ and $V_{k+2}$ are adjacent
because any vertex of $V_{k+2}$ is still in the distance at most $1-\eps/2$
from any vertex of $V_k$.
If boxes $V_k$ and $V_{k+2}$ were too close each to other then it
would contradict that $v_0,v_1,v_2,\ldots,v_d$ is the shortest path in $G$.
If $v\in V_{k+1}$, then we note that any vertex of $V_k$
is in the distance at most $1-\eps/2$ from
any vertex in the distance at most $1-\eps$ from $v$.
Thus we can define $W$ to be the set of the last boxes of
the modified sequences of boxes.
\qed
}

We should estimate what is the average maximal number of
stations transmitting in the box $C(i,j)$.

\begin{fact}\labell{fact23}
The expected value of the maximal number of
stations transmitting in the box $C(i,j)$ in round K1
of during one call of the loop ``{\bf for} $k$''
is at most 6.
\end{fact}

\def\FACTTWENTYTHREE{
\noindent
{\bf\em Proof of Fact~\ref{fact23}:}
There are two cases.
The first case is when the stations of the box $C(i,j)$
get silent after K3 when one or more stations transmit in round K1.
The second case is if they get silent because of the external noise.
In the second case the expected maximal number of stations is zero.
So we concentrate on the first case.

In the first case this expected number is
\[
E_1=\sum_{k=1}^{\log n} P_k \cdot E(\#\text{ transmitting nodes in $k$-th round K1}),
\]
where
\[
P_k=\Pr(\text{in first $k-1$ rounds K1 none transmits}).
\]
Denote the number of stations in the box $C(i,j)$ by $\Delta$.
Let $l=\left\lceil\log \frac{n}{\Delta}\right\rceil$.
We get
{\setlength\arraycolsep{0.1em}
\begin{eqnarray*}
E_1 &\le& 1\cdot \Delta 2^l/n+(1-2^l/n)^\Delta\cdot\Delta 2^{l+1}/n+\\
    &&    +(1-2^l/n)^\Delta(1-2^{l+1}/n)^\Delta\Delta 2^{l+2}/n+\cdots
\end{eqnarray*}}
Finally using the inequality $(1-1/\Delta)^\Delta\le 1/2$ we can estimate this sum as follows
\[E_1\le 2+2+1+1/2+1/4+1/8+\cdots\le 6.
\]
\qed
}

\begin{fact}\labell{fact24}
The probability, that in one call of the loop ``{\bf for} $k$''
the leader of the box $C(i,j)$ is elected is at least $1/18$.
\end{fact}

\def\FACTTWENTYFOUR{
\noindent
{\bf\em Proof of Fact~\ref{fact24}:}
A necessary condition for the successful leader election is that
exactly one station from $C(i,j)$ transmits while the total
interference in all the boxes in the distance at most 1 from the box $C(i,j)$
is at most ${\cal N}\alpha\eps/2$.
By the previous Fact the expected number of the maximum number of transmitting stations for each
$C(i,j)$ and related octant can be at most 7.
This value is attained in the round when all the box $C(i,j)$ switches off.
One of these stations ($u$) is in the octant and the rest in $C(i,j)$.
These all stations are situated in a square with edge length $\gamma\lfloor 1/\gamma\rfloor$.
All these squares are boxes of the grid $G_{\gamma\lfloor 1/\gamma\rfloor}$.
By the Markov inequality and Corollary \ref{corollary03} we have
\[
\Pr(\text{max interference $\ge {\cal N}\alpha\eps/2$})\le \frac{E(\text{max interference})}{{\cal N}\alpha\eps/2}\le \frac{1}{2}.
\]
Let $x\in\mathbb{R}$.
We estimate for $k=\left\lfloor x\log\frac{n}{\Delta}\right\rfloor$
the probability that exactly one node of $C(i,j)$ transmits $k$-th round K1
(while no station transmits in earlier rounds K1). We use the inequality $(1-x)(1-y)\ge 1-x-y$.
\[
P_2=\Delta \frac{2^k}{n} \left(1-\frac{2^k}{n}\right)^{\Delta-1}\prod_{l<k} \left(1-\frac{2^l}{n}\right)^\Delta\ge
      x'\left(1-2x'\right),
\]
where $x/2<x'=2^k\Delta/n\le x$.
If $x=1/3$, then $P_2\ge 1/9$.
This means that the probability that the leader is
elected is at least
\[
\Pr(\text{max interference}\le{\cal N}\alpha\eps/2)\cdot P_2\ge \frac{1}{18}.
\]
\qed
}

\begin{theorem}\labell{theorem25}
Algorithm RandUnknownBroadcast$(d,T)$ 
accomplishes broadcast in
$O(\bar{d}^2(D+\log(1/\delta))\log n)$ rounds,
with probability $1-\delta$, 
when run for $d= d_{\alpha,{\cal N}\alpha\eps/2,\gamma},
\bar{d}=\left\lfloor 1/\gamma\right\rfloor d_{\alpha,{\cal N}\alpha\eps/28,\gamma\lfloor 1/\gamma\rfloor}$
and for some $T=O(D+\log(1/\delta))$.
\end{theorem}

\begin{proof}
A necessary condition for the broadcast is
that each box of $W'$ obtains the message and broadcasts
it at least once to all stations in the range $1-\eps/2$.
Such a box $V$ in $W'$ transmits successfully when
the message is successfully transmitted at most $D$ times
on the shortest sequence of boxes from the source to $V$ and
finally is successfully transmitted by the box $V$.
The sufficient condition for this to happen is that a chain of
altogether at most $D$ successful leader elections happen.
The probability of such a successful leader election is by Fact~\ref{fact24}
bigger than $p=1/18$.

Now we estimate the probability, that our algorithm
completes the broadcast. Let the number of repetitions of the most external loop be
$t=2D/p + 2\ln(1/\delta')/p$ for some $\delta'\in\mathbb{R}$.
By Lemma~\ref{lemma13},
\vspace*{-1ex}
\begin{eqnarray*}
\lefteqn{\Pr(\text{some $\in W$ don't transmit successfully})\le}\\
   &&\sum_{V\in W} \Pr(\text{box $V$ doesn't transmit successfully})
\ .
\end{eqnarray*}
Therefore,
\vspace*{-1ex}
\[
\Pr(\text{some $V\in W$ don't transmit successfully})\le 4(D+1)^3\delta'
\ .
\]

\vspace*{-1ex}
\noindent
To get this probability smaller than $\delta$ we need
the number of repetitions of the most external loop
\[
T=\frac{2D}p + \frac{2\ln(1/\delta')}p +\frac{2\ln(4(D+1))}p=O(D+\log(1/\delta))
\ .
\]
Each run of the most external loop takes $O(\bar{d}^2\log n)$
rounds, 
which yields $O(\bar{d}^2(D+\log(1/\delta))\log n)$ rounds in total.
\end{proof}

%% file: bibliography-new.bbl

%% file: RandProofs.tex
\section{Omitted proofs from Section~\ref{s:dknown}}

\FACTONE

\vspace*{2ex}

\LEMMATWO

\vspace*{2ex}

\FACTELEVEN

\vspace*{2ex}

\FACTTWELVE

\vspace*{2ex}

\LEMMATHIRTEEN

\vspace*{2ex}

\FACTFOURTEEN

\section{Omitted proofs from Section~\ref{s:dunknown}}

\LEMMATWENTYTWO

\vspace*{2ex}

\FACTTWENTYTHREE

\vspace*{2ex}

\FACTTWENTYFOUR

%% file: SINRBroadcastWithNoise-arxiv-02-2013.bbl
\begin{thebibliography}{10}
\providecommand{\url}[1]{#1}
\csname url@samestyle\endcsname
\providecommand{\newblock}{\relax}
\providecommand{\bibinfo}[2]{#2}
\providecommand{\BIBentrySTDinterwordspacing}{\spaceskip=0pt\relax}
\providecommand{\BIBentryALTinterwordstretchfactor}{4}
\providecommand{\BIBentryALTinterwordspacing}{\spaceskip=\fontdimen2\font plus
\BIBentryALTinterwordstretchfactor\fontdimen3\font minus
  \fontdimen4\font\relax}
\providecommand{\BIBforeignlanguage}[2]{{%
\expandafter\ifx\csname l@#1\endcsname\relax
\typeout{** WARNING: IEEEtran.bst: No hyphenation pattern has been}%
\typeout{** loaded for the language `#1'. Using the pattern for}%
\typeout{** the default language instead.}%
\else
\language=\csname l@#1\endcsname
\fi
#2}}
\providecommand{\BIBdecl}{\relax}
\BIBdecl


\bibitem{YuHWTL12}
D.~Yu, Q.-S. Hua, Y.~Wang, H.~Tan, and F.~C.~M. Lau, ``Distributed
  multiple-message broadcast in wireless ad-hoc networks under the sinr
  model,'' in \emph{SIROCCO},
2012, pp. 111--122.

\dogory
\bibitem{YuWHL11}
D.~Yu, Y.~Wang, Q.-S. Hua, and F.~C.~M. Lau, ``Distributed local broadcasting
  algorithms in the physical interference model,'' in \emph{DCOSS}.\hskip 1em
  plus 0.5em minus 0.4em\relax IEEE, 2011, pp. 1--8.

\dogory
\bibitem{KV10}
T.~Kesselheim and B.~V{\"o}cking, ``Distributed contention resolution in
  wireless networks,'' in \emph{DISC},
2010, pp. 163--178.

\dogory
\bibitem{GoussevskaiaMW08}
O.~Goussevskaia, T.~Moscibroda, and R.~Wattenhofer, ``Local broadcasting in the
  physical interference model,'' in \emph{DIALM-POMC}, M.~Segal and
  A.~Kesselman, Eds.\hskip 1em plus 0.5em minus 0.4em\relax ACM, 2008, pp.
  35--44.

\dogory
\bibitem{WatSurv}
O.~Goussevskaia, Y.~A. Pignolet, and R.~Wattenhofer, ``Efficiency of wireless
  networks: Approximation algorithms for the physical interference model,''
  \emph{Foundations and Trends in Networking}, vol.~4, no.~3, pp. 313--420,
  2010.

\dogory
\bibitem{RichaSSZ}
\BIBentryALTinterwordspacing
A.~Richa, C.~Scheideler, S.~Schmid, and J.~Zhang, ``Towards jamming-resistant
  and competitive medium access in the sinr model,'' in \emph{Proceedings of
  the 3rd ACM workshop on Wireless of the students, by the students, for the
  students},
2011, pp. 33--36.
\BIBentrySTDinterwordspacing

\dogory
\bibitem{DessmarkP07}
A.~Dessmark and A.~Pelc, ``Broadcasting in geometric radio networks,'' \emph{J.
  Discrete Algorithms}, vol.~5, no.~1, pp. 187--201, 2007.

\dogory
\bibitem{EmekGKPPS09}
Y.~Emek, L.~Gasieniec, E.~Kantor, A.~Pelc, D.~Peleg, and C.~Su, ``Broadcasting
  in udg radio networks with unknown topology,'' \emph{Distributed Computing},
  vol.~21, no.~5, pp. 331--351, 2009.

\dogory
\bibitem{EmekKP08}
Y.~Emek, E.~Kantor, and D.~Peleg, ``On the effect of the deployment setting on
  broadcasting in euclidean radio networks,'' in \emph{PODC}, R.~A. Bazzi and
  B.~Patt-Shamir, Eds.\hskip 1em plus 0.5em minus 0.4em\relax ACM, 2008, pp.
  223--232.

\dogory
\bibitem{GasieniecKLW08}
L.~Gasieniec, D.~R. Kowalski, A.~Lingas, and M.~Wahlen, ``Efficient
  broadcasting in known geometric radio networks with non-uniform ranges,'' in
  \emph{DISC},
2008, pp. 274--288.

\dogory
\bibitem{SenH96}
A.~Sen and M.~L. Huson, ``A new model for scheduling packet radio networks,''
  in \emph{INFOCOM}, 1996, pp. 1116--1124.

\dogory
\bibitem{Kow-PODC-05}
D.~R. Kowalski, ``On selection problem in radio networks,'' in \emph{PODC},
  M.~K. Aguilera and J.~Aspnes, Eds.\hskip 1em plus 0.5em minus 0.4em\relax
  ACM, 2005, pp. 158--166.

\dogory
\bibitem{KP-DC-05}
D.~R. Kowalski and A.~Pelc, ``Broadcasting in undirected ad hoc radio
  networks,'' \emph{Distributed Computing}, vol.~18, no.~1, pp. 43--57, 2005.

\dogory
\bibitem{DeMarco-SICOMP-10}
G.~DeMarco, ``Distributed broadcast in unknown radio networks,'' \emph{SIAM J.
  Comput.}, vol.~39, no.~6, pp. 2162--2175, 2010.

\dogory
\bibitem{Censor-HillelGKLN11}
K.~Censor-Hillel, S.~Gilbert, F.~Kuhn, N.A. Lynch, C.C. Newport,
  ``Structuring unreliable radio networks,'' in \emph{PODC},
2011, pp. 79--88.

\dogory
\bibitem{KushilevitzM98}
E.~Kushilevitz and Y.~Mansour, ``An omega({\it d} log ({\it n/d})) lower bound
  for broadcast in radio networks,'' \emph{SIAM J. Comput.}, vol.~27, no.~3,
  pp. 702--712, 1998.

\dogory
\bibitem{CzumajRytter-FOCS-03}
A.~Czumaj and W.~Rytter, ``Broadcasting algorithms in radio networks with
  unknown topology,'' in \emph{FOCS},
2003, pp. 492--501.

\dogory
\bibitem{Farach-ColtonM07}
M.~Farach-Colton and M.~A. Mosteiro, ``Sensor network gossiping or how to break
  the broadcast lower bound,'' in \emph{ISAAC},
2007, pp. 232--243.

\dogory
\bibitem{ClementiMS01}
A.~E.~F. Clementi, A.~Monti, and R.~Silvestri, ``Selective families,
  superimposed codes, and broadcasting on unknown radio networks,'' in
  \emph{SODA},
  ACM/SIAM, 2001, pp. 709--718.

\end{thebibliography}
